\documentclass[letterpaper, 10 pt, conference]{ieeeconf}  

\IEEEoverridecommandlockouts                              

\usepackage[utf8]{inputenc}
\usepackage[T1]{fontenc}
\usepackage{amsmath}
\usepackage{amsfonts}
\usepackage{amssymb}
\usepackage{enumerate}
\let\labelindent\relax
\usepackage[shortlabels]{enumitem}
\usepackage{graphicx}
\usepackage{import}
\usepackage{subfigure}
\usepackage{xcolor}
\usepackage{easyReview}
\usepackage[ruled,noend]{algorithm2e} 
\usepackage{balance}

\usepackage{soul}

\newcommand{\bbN}{\mathbb{N}}

\newcommand{\bbR}{\mathbb{R}}

\newcommand{\calA}{\mathcal{A}}
\newcommand{\calB}{\mathcal{B}}
\newcommand{\calC}{\mathcal{C}}

\newcommand{\calF}{\mathcal{F}}
\newcommand{\calG}{\mathcal{G}}
\newcommand{\calH}{\mathcal{H}}
\newcommand{\calI}{\mathcal{I}}

\newcommand{\calK}{\mathcal{K}}

\newcommand{\calQ}{\mathcal{Q}}

\newcommand{\calS}{\mathcal{S}}
\newcommand{\calT}{\mathcal{T}}
\newcommand{\calU}{\mathcal{U}}

\newcommand{\calX}{\mathcal{X}}

\usepackage{amsthm}

\theoremstyle{definition}
\newtheorem{assumption}{Assumption}
\newtheorem{theorem}{Theorem}
\newtheorem{lemma}[theorem]{Lemma}
\newtheorem{proposition}[theorem]{Proposition}
\newtheorem{corollary}[theorem]{Corollary} 
\newtheorem{definition}{Definition}

\newtheorem{example}{Example}

\theoremstyle{remark}
\newtheorem{remark}{Remark}

\title{\LARGE \bf
Handling Disjunctions in Signal Temporal Logic Based Control Through Nonsmooth Barrier Functions
}

\author{Adrian Wiltz and Dimos V. Dimarogonas
\thanks{This work was supported by the ERC Consolidator Grant LEAFHOUND, the Swedish Foundation for Strategic Research (SSF) COIN, the Swedish Research Council (VR) and the Knut och Alice Wallenberg Foundation.}
\thanks{The authors are with the Division
	of Decision and Control Systems, KTH Royal Institute of Technology, SE-100 44 Stockholm, Sweden \{wiltz, dimos\}@kth.se.}%
\thanks{This work has been accepted for presentation at the 61st IEEE Conference on Decision and Control 2022, in Cancún, Mexico. Proofs have been subject to the review process but are not included in the version presented at the conference due to space limitations.}
}

\begin{document}

\maketitle
\thispagestyle{empty}
\pagestyle{empty}

\setlength{\abovedisplayskip}{0.1cm}
\setlength{\belowdisplayskip}{0.1cm}

\begin{abstract}

For a class of spatio-temporal tasks defined by a fragment of Signal Temporal Logic (STL), we construct a nonsmooth time-varying control barrier function (CBF) and develop a controller based on a set of simple optimization problems. Each of the optimization problems invokes constraints that allow to exploit the piece-wise smoothness of the CBF for optimization additionally to the common gradient constraint in the context of CBFs. In this way, the conservativeness of the control approach is reduced in those points where the CBF is nonsmooth. Thereby, nonsmooth CBFs become applicable to time-varying control tasks. Moreover, we overcome the problem of vanishing gradients for the considered class of constraints which allows us to consider more complex tasks including disjunctions compared to approaches based on smooth CBFs. As a well-established and systematic method to encode spatio-temporal constraints, we define the class of tasks under consideration as an STL-fragment. The results are demonstrated in a relevant simulation example.

\end{abstract}

\section{Introduction}

In applications, one often encounters spatio-temporal constraints which impose both state- and time-constraints on a system. Logic expressions can be used to express such constraints. For example, one can form out of elementary rules \emph{Robot 1 must move within 5 seconds to region $ A $} (R1), or \emph{Robot 2 must move within 5 seconds to region $ B $} (R2), and \emph{Robot 1 and Robot 2 must keep a distance of at most $ d $ to each other} (R3) the overall rule $ (\text{R1} \vee \text{R2}) \wedge \text{R3} $ where $ \wedge, \vee $ denote logic AND and OR, respectively. Temporal logics like
STL (Signal Temporal Logic) \cite{Maler2004} allow the specification of such spatio-temporal constraints and increase the expressiveness of boolean logic by the temporal aspect. In the sequel, we call a composition of various spatio-temporal constraints by logic operators a \emph{task}. Although STL originates from the field of formal verification in computer science, it is becoming increasingly popular as a well-established and systematic method to formulate spatio-temporal tasks in the field of control. Therefore, we also define the class of tasks under consideration as an STL-fragment in this paper.  Most available control approaches for spatio-temporal tasks as  \cite{
Belta2007, Fainekos2005, Loizou2004} are based upon automata theory, which is often computationally expensive due to state discretization. Thereby, potential field based methods can be a computationally efficient alternative for some classes of spatio-temporal constraints~\cite{Lindemann2019,Lindemann2020a}.

Potential field based methods have a long tradition in control theory and have been successfully applied to tasks as collision- \cite{Dimarogonas2006} and obstacle-avoidance \cite{Panagou2013} as well as spatio-temporal tasks~\cite{Lindemann2019,Lindemann2020a}. In the latter case, smooth time-varying Control Barrier Functions (CBF) are employed. CBFs, introduced in~\cite{Prajna2004} and~\cite{Wieland2007}, are a control concept for ensuring the invariance of sets and proved to be a suitable tool for guaranteeing the satisfaction of state constraints on control problems~\cite{Ames2019}. 
By now, a broad range of results on CBFs can be found in the literature. Whereas first approaches on CBFs \cite{Wieland2007,Ames2014} consider systems with relative degree one, \cite{Nguyen2016,Xiao2021b} also consider systems with higher relative degree. 

With view to constraints specified via logic expressions in the context of CBFs, especially two approaches must be named: For state-constraints specified via boolean logic, \cite{Glotfelter2020}~employs nonsmooth CBFs. On the other hand, \cite{Lindemann2019} constructs a smooth time-varying CBF that ensures the satisfaction of specified spatio-temporal tasks defined via an STL-fragment. However, since \cite{Lindemann2019} uses a smoothed approximation of  maximum and minimum operators, there exist points where the gradient of the CBF vanishes. This may be problematic when considering disjunctions (logic OR) in the context of time-varying CBFs.

In this paper, we resolve this problem by using a non-smooth CBF approach and can therefore take also disjunctions into account when considering spatio-temporal tasks. 
In contrast to~\cite{Glotfelter2020}, it is too conservative to require that a control action results in an ascend on multiple ``active'' CBFs at the same time. 
Therefore, although inspired by~\cite{Filippov1964,Filippov1988}, we do not base our control approach on the Filippov-operator and differential inclusions as \cite{Glotfelter2020}, and employ a somewhat different approach in those points where the CBF is nonsmooth. In fact, we can circumvent the usage of differential inclusions by basing our controller on a set of  optimization problems that exploit the piecewise smoothness of the CBF and we can show that the solutions to the closed-loop system are Carathéodory solutions. Thereby, we make the nonsmooth approach less conservative and thus applicable to time-varying CBFs.

The sequel is structured as follows: Section~\ref{sec:preliminaries} introduces the considered dynamics, nonsmooth time-varying CBFs, and reviews STL; Section~\ref{sec:main results} constructs a nonsmooth CBF candidate for the STL-fragment under consideration, presents the control approach and proves set invariance; Section~\ref{sec:simulations} presents a relevant simulation example and demonstrates applicability of the proposed control scheme; Section~\ref{sec:conclusion} summarizes the conclusions of this paper. All proofs to the derived theoretic results can be found in the appendix.



\paragraph*{Notation} Sets are denoted by calligraphic letters. Let $ \calA\subseteq\bbR^n,\calB\subseteq\bbR^{m} $, and let $ d(\cdot,\cdot) $ define a metric on $ \calA $. The $ \varepsilon $-neighborhood of $ x\in\calA $ is $ B_{\varepsilon}(x) := \lbrace y\in \calA \,|\, d(x,y) < \varepsilon \rbrace $, $ \text{Int}\,\calA $ the interior, $ \partial\calA $ the boundary of $ \calA $; the Lebesgue measure of $ \calA'\subseteq \calA $ is $ \mu(\calA') $, and if a property of a function $ f:\calA\rightarrow\calB $ holds everywhere on $ \calA\smallsetminus\calA' $ with $ \mu(\calA') = 0 $, we say that it holds almost everywhere (a.e.). Let $ \calI\subset\bbN $ be a finite index set with cardinality $ |\calI| $ and $ \lbrace a_i \rbrace_{i\in\calI} := \lbrace a_i \,|\, i\in\calI \rbrace $. Let $ f_i: \calA\rightarrow\calB $. The maximum and minimum operators are denoted by $ \min_{i\in\calI} f_i $ and $ \max_{i\in\calI} f_i $, respectively, and we define $ \min_{i\in\calI} f_i := 0 $, $ \max_{i\in\calI} f_i := 0 $ for $ \calI = \emptyset $. A function $ \alpha: \bbR_{\geq0} \rightarrow \bbR_{\geq0} $ is a class $ \calK $ function if it is continuous, strictly increasing and $ \alpha(0)=0 $. The left and right sided derivatives of a function $ f(t) $ with respect to $ t $ where $ f:\bbR\rightarrow\bbR $ are defined as $ d_{t^{-}} f(t) := \lim_{\nu\rightarrow 0^{-}} \frac{f(t+\nu)-f(t)}{\nu} $ and $ d_{t^{+}} f(t) := \lim_{\nu\rightarrow 0^{+}} \frac{f(t+\nu)-f(t)}{\nu} $, respectively. For a vector-field $ g $ and a smooth real-valued scalar function $ h $, we denote the Lie-derivative by $ L_{g}h $. The inverse unit step is $ \sigma^{-1}(x) := $  {\tiny $ \begin{cases} 1 & x \leq 0 \\0 & x > 0 \end{cases} $}. Logic \emph{and} and \emph{or} are denoted by $ \wedge $ and $ \vee $, respectively. 


\section{Preliminaries}
\label{sec:preliminaries}

At first, we introduce the system dynamics under consideration, redefine the concept of control barrier functions (CBF) in order to suit the control problem, and review STL-formulas as a formalism for defining complex spatio-temporal constraints on a control problem. 

\vspace{-0.1cm}
\subsection{System Dynamics}

We consider the input-affine system
\begin{align}
\label{eq:input affine system}
\dot{x} = f(x) + g(x) u, \quad x(t_0) = x_0
\end{align}
on the closed time-interval $ \calT = [t_1,t_2] \subseteq \bbR $, $ t_0\in\calT $, where $ x\in\calX \subseteq\bbR^{n} $, $ u \in\bbR^{m} $, and $ f,g $ are continuous functions with respective dimensions. 
Besides, we say that a time-varying set $ \calC(t) \subseteq \calX $ is \emph{forward time-invariant} for system~\eqref{eq:input affine system}, if $ x(t)\in\calC(t) \; \forall t\geq t_0 $ for $ x(t_0) = x_0 \in\calC(t_0) $. 

\vspace{-0.1cm}
\subsection{Non-Smooth Time-Varying Control Barrier Functions}
\label{subseq:non-smooth time-varying CBFs}
Let $ b(t,x) $ be a real-valued function $ b: \calT \times \calX \rightarrow \bbR $, and we define $ T \in \calT $ as the time for which $ b(t,x) \equiv 0 $ for all $ t > T $\footnote{If there is no $ t\in\calT $ such that $ b(t,x) \equiv 0 $ holds, then we can still set $ T = t_2 $.}. We assume that $ b(t,x) $ is continuous and piecewise continuously differentiable in~$ x $, almost everywhere continuously differentiable in~$ t $, and  
\begin{align}
\label{eq:continuity condition on b}
\lim_{\tau\rightarrow t^{-}} b(\tau,x) < \lim_{\tau\rightarrow t^{+}} b(\tau,x) 
\end{align}
holds for discontinuities at $ t < T $. Due to their role in what follows, we call a function $ b(t,x) $ with the aforementioned properties a \emph{barrier function}~(BF).
 
\vspace{-0.1cm}
\begin{definition}[Safe Set]
	\label{def:safe set}
	Let $ \calT $ be the time-interval where \eqref{eq:input affine system} is defined. The set-valued function $ \calC(t) := \lbrace x\in\calX \, | \, b(t,x) \geq 0 \rbrace $ is called a time-varying \emph{safe set} where $ b $ is a barrier function. 
	If $ x(t) \in \calC(t) $ for all times $ t\in\calT $, we call $ x $ \emph{feasible}.
\end{definition}
\vspace{-0.1cm}

Note that $ \calC(t) \equiv \calX $ for $ t > T $, i.e., the condition $ b(t,x) \geq 0 $ is trivially satisfied as $ b(t,x) \equiv 0 $. Next, we derive from the continuity properties of $ b(t,x) $ the following continuity properties of the set-valued function~$ \calC(t) $.

\vspace{-0.15cm}
\begin{lemma}
	\label{lemma:continuity properties of C}
	The time-varying superlevel sets $ \calC'(t) := \lbrace x\in\calX \, | \, b(t,x) \geq c \rbrace $, $ c\in\bbR $, of $ b(t,x) $ are continuous a.e. with respect to $ t $, and at a discontinuity at time $ t $ it holds 
	\begin{align}
	\label{eq:continuity condition of C}
	\lim_{\tau\rightarrow t^{-}}\calC'(\tau) \subseteq \lim_{\tau\rightarrow t^{+}}\calC'(\tau).
	\end{align}
\end{lemma}
\vspace{-0.15cm}

For details on the continuity of set-valued functions, we refer to \cite[Ch.~5B]{Rockafellar2009}. By~\eqref{eq:continuity condition of C} it is ensured that for a discontinuity at time $ t $ it holds $ x\in\lim_{\tau \rightarrow t^{-}} \calC(\tau) \Rightarrow x\in\lim_{\tau\rightarrow t^{+}}\calC(\tau) $, i.e., a feasible state stays feasible. Finally, we define control barrier functions for the nonsmooth time-varying case as follows. 

\vspace{-0.15cm}
\begin{definition}[Control Barrier Function (CBF)]
	A barrier function $ b(t,x) $ is a \emph{control barrier function} for system~\eqref{eq:input affine system} if there exists a class $ \calK $ function~$ \alpha $ such that 
	\begin{align}
	\label{eq:cbf condition}
	\sup_{u} d_{\delta^{+}} \, b(t\!+\!\delta, x\!+\!\delta (f(x)\!+\!g(x)u)) \bigg|_{\delta = 0} \!\! \geq -\alpha(b(t,x))
	\end{align}
	for all $ x\!\in\!\calC(t) $ and all $ t\!\in\!\calT $ where $ b $ is continuous.
\end{definition}
\vspace{-0.15cm}

\begin{remark}
	The derivative on the left is a right sided directional derivative similar to~\cite[p.~155]{Filippov1988}. The advantage of this formulation is that it is well-defined for Lipschitz-continuous functions $ b $ which are not necessarily everywhere differentiable. As we are concerned with forward invariance, we consider the right-sided directional derivative. In the proof of Proposition~\ref{prop:equivalence optimization problems}, we relate~\eqref{eq:cbf condition} to other commonly used CBF gradient conditions as the one in~\cite{Lindemann2019}.
\end{remark}
\vspace{-0.15cm}


\vspace{-0.1cm}
\subsection{Signal Temporal Logic (STL)}
Next, we briefly review Signal Temporal Logic (STL), and specify the considered class of tasks as an STL-fragment. STL is a predicate logic with temporal operators. A predicate $ p $ has a truth value which is defined by 
$ p := $  {\tiny $
\begin{cases}
\top & \text{if } h(x)\geq 0 \\
\bot & \text{if } h(x)< 0
\end{cases} $}
where $ \top $ and $ \bot $ denote \emph{True} and \emph{False}, respectively, and $ h: \calX \rightarrow \bbR $ a predicate function. The grammar of a general STL formula is given as \cite{Maler2004}
\begin{align}
\label{eq:STL formula}
\theta ::= \top | p | \neg \theta | \theta_1 \vee \theta_2 | \theta_1\calU_{[a,b]}\theta_2
\end{align} 
where $ \theta_1, \theta_2 $ are STL formulas, and $ 0\leq a \leq b $ with $ a,b \in \bbR_{\geq 0} $. The satisfaction relation $ (x,t) \vDash \theta $ indicates that a time-dependent function $ x $ satisfies $ \theta $ from time $ t $ onwards. It is inductively defined as 
\begin{subequations}
	\label{eq:STL grammar}
	\begin{align}
		\label{seq:STL grammar predicates}
		&(x,t) \vDash p &&\Leftrightarrow h(x(t)) \geq 0 \\
		&(x,t) \vDash \neg \theta &&\Leftrightarrow \neg((x,t)\vDash \theta) \Leftrightarrow (x,t) \nvDash \theta \\
		\label{seq:STL or}
		&(x,t) \vDash \theta_1 \vee \theta_2 &&\Leftrightarrow (x,t) \vDash \theta_1 \text{ or } (x,t) \vDash \theta_2 \\
		&(x,t) \vDash \theta_1 \calU_{[a,b]} \theta_2 \!\!\!\!\!\!\!\!&& \Leftrightarrow
		\begin{aligned}[t]
			& \exists t' \!\!\in \!\![t\!+\!a,t\!+\!b] \,\text{s.t.}\, (s,t')\vDash \theta_2 \\ 
			& \text{ and } (x,t'') \vDash \theta_1, \, \forall t''\in [t,t'].
		\end{aligned}
	\end{align} 
\end{subequations}
Instead of $ (x,0)\vDash \theta $, we also write $ x \vDash \theta $ in the sequel. Due to De Morgans law, this grammar also includes conjunctions, defined as $ \theta_1 \wedge \theta_2 := \neg(\neg\theta_1 \vee \neg\theta_2) $. Moreover starting with the \emph{until} operator, the \emph{eventually} and \emph{always} operators can be defined as $ (x,t) \vDash \calF_{[a,b]} \theta := \top \calU_{[a,b]}\theta $ and $ (x,t) \vDash \calG_{[a,b]} \theta := \neg\calF_{[a,b]} \neg\theta $, respectively, and it equivalently holds
\begin{subequations}
	\label{eq:STL grammar eventually and always}
	\begin{align}
		\label{seq:STL grammar eventually}
		(x,t) \vDash \calF_{[a,b]} \theta  &\Leftrightarrow \exists t' \in [t+a,t+b] \text{ s.t. } (x,t') \vDash \theta, \\
		\label{seq:STL grammar always}
		(x,t) \vDash \calG_{[a,b]} \theta  &\Leftrightarrow (x,t') \vDash \theta, \, \forall t'\in[t+a,t+b].
	\end{align}
\end{subequations}

In the sequel, we consider the STL-fragment
\begin{subequations}
	\label{eq:STL fragment}
	\begin{align}
	\label{seq:STL fragment psi}
	\psi &::= \top | p | \psi_1 \vee \psi_2 | \psi_1 \wedge \psi_2 \\
	\label{seq:STL fragment phi}
	\phi &::= \phi_1 \vee \phi_2 | \phi_1 \wedge \phi_2 |\calF_{[a,b]}\psi | \calG_{[a,b]}\psi | \psi_1 \calU_{[a,b]}\psi_2
	\end{align}
\end{subequations}
where $ a,b\in\calT $ and $ a\leq b $, which is the fragment considered in~\cite{Lindemann2019} extended by disjunctions. Below we denote STL-formulas satisfying grammar~\eqref{seq:STL fragment psi} or~\eqref{seq:STL fragment phi} by $ \psi_i $ or $ \phi_{i} $, respectively. As we see later, the problem of vanishing gradients in the presence of disjunctions as encountered in~\cite{Lindemann2019,Lindemann2020a} can be resolved with a non-smooth approach. 

\vspace{-0.15cm}
\begin{assumption}
	\label{ass:predicate function}
	We assume that $ h: \calX \rightarrow \bbR $ is a continuously differentiable and concave function. 
	Let $ {\calH} $ be the set of all maximum points of $ h(x) $, i.e., $ \calH := \lbrace x_0\in\calX \, | \,  \exists \varepsilon > 0: ||x_0 - x|| < \varepsilon \Rightarrow h(x) < h(x_0) \rbrace $. We additionally assume that $ L_{g}h(x) \neq 0 \; \forall x\in\calX\smallsetminus\calH$ and $ L_{g}h(x) \neq 0 $ if $ L_{f}h(x) \neq 0 $ for $ x\in\calH $ (first-order condition on $ h $).
\end{assumption}



\vspace{-0.15cm}
\section{Main Results}
\label{sec:main results}
\vspace{-0.1cm}

In the sequel, we present the construction of a BF which parallels~\cite{Glotfelter2020,Lindemann2019} in parts, and show that it satisfies the properties assumed in Section~\ref{subseq:non-smooth time-varying CBFs}. Thereafter, we outline the proposed control approach and prove the invariance of safe sets for the closed-loop system.
\vspace{-0.2cm}

\subsection{Construction of BFs}
\label{subsec:construction of candiate CBFs}
\vspace{-0.1cm}

Consider an STL-formula $ \phi_{0} $ that satisfies grammar~\eqref{eq:STL fragment} and comprises predicates $ \lbrace p_{i} \rbrace_{i\in\calI^{e}} $ where $ \calI^{e}\subset\bbN $ is an index set; the corresponding predicate functions are $ \lbrace h_{i}(x) \rbrace_{i\in\calI^{e}} $. In this section, our goal is to construct a BF $ b_0(t,x) $ for the STL-formula $ \phi_0 $ such that $ b_{0}(t,x(t)) \geq 0 \; \forall t\in\calT $ implies $ x \vDash \phi_0 $; then we say that $ b_0 $ \emph{implements} the STL-formula~$ \phi_0 $. 

In a first step, we construct the BFs $ \lbrace b_{i} \rbrace_{i\in\calI^{e}} $  for each of the predicates $ \lbrace p_{i} \rbrace_{i\in\calI^{e}} $ which we call \emph{elementary barrier functions}:

\begin{enumerate}[start=0,label={R\arabic*:}]
	\item For $ \psi_i \!=\! p_{i} $, the corresponding BFs are $ b_{i}(t,x) \!:=\! h_i(x) $.
\end{enumerate}

Using the set of elementary BFs as a starting point, we can recursively construct a BF $ b_0 $ implementing $ \phi_0 $. Therefore, we introduce rules for the construction of BFs $ b_{i} $ which implement STL-formulas $ \psi_i $ and $ \phi_i $ satisfying grammar~\eqref{seq:STL fragment psi} and~\eqref{seq:STL fragment phi}, respectively,  as a composition of already constructed BFs $ \lbrace b_{i'} \rbrace_{i'\in\calB_i} $ where $ \calB_{i} \subset \bbN $ is a finite index set. The construction rules are given as follows:

\begin{enumerate}[start=1,label={R\arabic*:}]
	\item If $ \psi_i = \bigwedge_{i'\in\calB_i} \psi_{i'} $, choose $ b_{i}(t,x) = \min_{i'\in\calB_i} b_{i'}(t,x) $.
	\item If $ \psi_i = \bigvee_{i'\in\calB_i} \psi_{i'} $, choose $ b_{i}(t,x) \!=\! \max_{i'\in\calB_i} b_{i'}(t,x) $. 
	\item For $ \phi_{i} = \calF_{[a,b]} \psi_{i'} $, we have $ \calB_i = \lbrace i' \rbrace $ and choose $ b_{i}(t,x) = (b_{i'}(t,x) + \gamma_{i}(t))\sigma^{-1}(t-\beta_{i}) $ where $ \sigma^{-1} $ is the inverse unit step as defined in the notation section, $ \gamma_{i}\!:\! \calT \!\!\rightarrow\! \bbR $ is a continuously differentiable function such~that $ \exists t' \!\in\! [a,b]\!: \gamma_{i}(t') \!\leq\! 0 $, and time $ \beta_{i} \!:=\! \min \lbrace t'  | \gamma_{i}(t') \!\leq\! 0 \rbrace $. 
	\item For $ \phi_{i} \!=\! \calG_{[a,b]} \psi_{i'} $, we have $ \calB_i \!=\! \lbrace i' \rbrace $ and choose $ b_{i}(t,x) \!=\! (b_{i'}(t,x) + \gamma_{i}(t))\sigma^{-1}(t-\beta_{i}) $ where $ \sigma^{-1} $~is the inverse unit step, $ \gamma_{i}\!\!:\! \calT \rightarrow \bbR $ is continuously differentiable, $ \gamma_{i}(t') \!\leq\! 0  $ for all $ t' \!\!\in\! [a,b] $, and time $ \beta_{i} \!\!:=\!\! b $. 
	\item If $ \phi_i \!=\! \bigwedge_{i'\in\calB_i} \phi_{i'} $, choose $ b_{i}(t,x) \!=\! \min_{i'\in\widetilde{\calB}_i(t)} b_{i'}(t,x) $ where $ \widetilde{\calB}_i(t) \!\!:=\!\! \lbrace i' \!\!\in\!\! \calB_i | t \!\leq\! \beta_{i'} \!\rbrace $ which means that $ b_{i'} $ is \emph{deactivated} at time $ \beta_{i'} $, i.e., $ b_{i} $ does not depend on $ b_{i'} $ for $ t>\beta_{i'} $ anymore. In addition, set time $ \beta_{i} := \max_{i'\in\calB_i} \beta_{i'} $.
	\item If $ \phi_i \!\!=\!\! \bigvee_{\!\!i'\!\in\!\calB_i} \!\phi_{i'}\! $, choose $ b_{i}(\!t,\!x\!) \!\!=\!\! \max_{i'\!\in\!\widetilde{\calB}_i\!(\!t\!)} \!b_{i'}\!(\!t,\!x\!)  \sigma^{\!-\!1}\!(\!t\!-\!\beta_{i}\!) $ with $ \widetilde{\calB}_i(t) \!=\! \lbrace i'\in\calB_i \, | \, b_{i'}(\tau,x(\tau)) \geq 0 \; \forall \tau \!\in\! [t_1,t] \rbrace $ and time $ \beta_{i} \!:=\! \min_{i'\in\widetilde{\calB}_i(t)} \beta_{i'} $.
	\item For $ \phi_{i} = \psi_{i'} \calU_{[a,b]} \psi_{i''} $, note that $ (x,t) \vDash \psi_{i'} \calU_{[a,b]} \psi_{i''} \Leftrightarrow \exists t'\in[t+a,t+b] \text{ s.t. } (x,t) \vDash \calG_{[0,t']} \psi_{i'} \wedge \calF_{[a,t']} \psi_{i''} $ and construction rules R2, R3 and R5 can be applied.
\end{enumerate}
For BFs $ b_i $ constructed in R0, we define $ \calB_{i} \!=\! \emptyset $; for $ b_{i} $ constructed in R0, R1, R2, R3 and R4, we set $ \gamma_{i}(t) \!\equiv\! 0 $ and $ \widetilde{\calB}_{i}(t)\!\equiv\!\calB_i $; and for $ b_{i} $ constructed in R0, R1 and R2, we set $ \beta_{i}\!=\!\infty $. We call a scalar $ \beta_{i} $ \emph{deactivation time} of~$ b_i $. Deactivation times reduce the conservativeness of $ b_{0} $, cf.~\cite{Lindemann2019}. The set containing the indices of all BFs $ b_i $ constructed in intermediate steps of the construction of $ b_{0} $ is denoted by $ \calI $ and must be distinguished from $ \calI^{e}\! $. Besides, we require that $ \calB_{i_{1}}\cap\calB_{i_{2}}\!=\!\emptyset $ for all $ i_{1},i_{2}\!\in\!\calI$, $ i_{1}\!\neq\! i_{2} $, i.e., every BF might be used for the construction of at most one other BF and hence $ b_{0} $ assumes a tree structure as illustrated by Figure~\ref{fig:CBF design}. 
We~illustrate the application of the construction rules R0-R7 with~an example.
\vspace{-0.15cm}
\begin{example}
	\label{ex:CBF construction}
	Consider $ \!\phi_0 \!\!=\!\! \calG_{\![a,b]} (h_{11}(x) \!\!\geq\!\! 0) \wedge \calF_{\![c,d]}( h_{211}(x) \!\geq\! 0 \vee h_{212}(x) \!\geq\! 0  ) $; here, we have $ \calI^{e} \!=\! \lbrace 11,211,212 \rbrace $. According to R0, define $ \psi_{11} \!:=\! p_{11} $, $ \psi_{211} \!:=\! p_{211} $, and $ \psi_{212} \!:=\! p_{212} $ with $ h_{11} $, $ h_{211} $, $ h_{212} $ as the respective predicate functions of predicates $ p_{11}, p_{211}, p_{212} $. Then, we start the recursive construction by choosing the elementary BFs as $ b_{11}(t,x) \!:=\! h_{11}(x) $, $ b_{211}(t,x) \!:=\! h_{211}(x) $, $ b_{212}(t,x) \!:=\! h_{212}(x) $. Moreover, we define $ \phi_1 \!:=\! \calG_{[a,b]}\psi_1 $, $ \psi_{21} \!:=\! \psi_{211} \vee \psi_{212} $, $ \phi_2 \!:=\! \calF_{[a,b]}\psi_{21} $, and construct $ b_{1} $, $ b_{21} $, $ b_{2} $ by applying R4, R2, R3, respectively. Since $ \phi_0 \!=\! \phi_1 \wedge \phi_2 $, we finally obtain $ b_0(t,x) \!=\! \min\lbrace h_{11}(x) \!+\! \gamma_1(t), \max\lbrace h_{211}(x), h_{212}(x) \rbrace \!+\! \gamma_2(t) \rbrace $ by applying R5. Besides, $ \calI \!=\! \lbrace 0, 1, 2, 11, 21, 211, 212 \rbrace $. Figure~\ref{fig:CBF design} illustrates the successive construction of $ b_0 $; the chosen indices emphasize the relation of the BFs among each other.
\end{example}
\vspace{-0.15cm}

\begin{figure}[btp]
	\centering
	\def\svgwidth{0.8\columnwidth}
	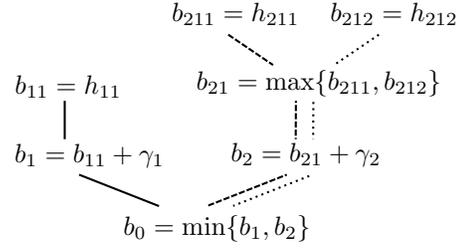
	\caption{Illustration of the recursive construction of $ b_0 $.}
	\label{fig:CBF design}
	\vspace{-0.65cm}
\end{figure}


Now, we show that $ b_{0} $ exhibits the assumed properties from Section~\ref{subseq:non-smooth time-varying CBFs} and thus constitutes a BF. 

\vspace{-0.15cm}
\begin{lemma}
\label{lemma:b is BF}
	The function $ b_{0}(t,x) $ is a BF.
\end{lemma}
\vspace{-0.15cm}

Next, we prove that the satisfaction of the time-dependent state-constraint $ b_{0}(t,x(t)) \geq 0 $ for all $ t\in\calT $ implies the satisfaction of the STL-formula $ \phi_0 $. In the next theorem, let $ b_{\psi} $ implement an STL-formula $ \psi $ satisfying grammar~\eqref{seq:STL fragment psi}, and $ b_{\phi} $ implement $ \phi $ satisfying grammar~\eqref{seq:STL fragment phi}. 

\vspace{-0.15cm}
\begin{theorem}
	\label{thm:task satisfaction}
	If $ b_{\phi}(t,x(t)) \geq 0 $ for all $ t\in\calT $, then $ x \vDash \phi $. Besides, $ b_{\psi}(t,x(t)) \geq 0 \Rightarrow (x,t) \vDash \psi $.
\end{theorem}
\vspace{-0.15cm}

In the sequel, we require the following assumption in addition to the fact that $ b_0 $ is a BF.

\vspace{-0.15cm}
\begin{assumption}
	\label{ass:ccbf}
	Let $ x $ be a maximum point of $ b_{0} $ at a given time~$ t $, i.e., $ b_{0}(t,x) \geq b_{0}(t,x') \; \forall x'\in B_{\varepsilon}(x) $ for some $ \varepsilon > 0 $. We assume that there exists a constant $ b_{\text{min}}\in\bbR $ such that $ b_{0}(t,x) > b_{\text{min}} > 0 $ for any maximum point $ x $ of $ b_0 $ at any given time $ t>T $. 
\end{assumption}
\vspace{-0.15cm}

\vspace{-0.15cm}
\begin{remark}
	Assumption~\ref{ass:ccbf} excludes STL-formulas that require predictions in order to ensure forward invariance. Such control tasks are beyond the scope of this paper. 
	In particular, Assumption~\ref{ass:ccbf} implies that there exist connected sets $ \calC_{i'}(t) $ such that $ \calC(t) \!=\! \bigcup_{i'} \calC_{i'}(t) $ where $ \text{Int}(\calC_{i'}(t))\!\neq\!\emptyset $ for all $ t\!\in\!\calT $.
\end{remark}
\vspace{-0.15cm}

In the next section, we present a control scheme based on the constructed BF $ b_0 $ and show that it is a CBF.

\subsection{Controller Design}
\label{subsec:controller design}

In related CBF literature \cite{Ames2019,Glotfelter2020,Lindemann2019}, an optimization problem is solved that ensures the satisfaction of a CBF gradient condition. However, directly solving
\begin{subequations}
	\label{eq:general optimization problem}
	\begin{align}
	&u^{\ast} = \underset{u}{\mathrm{argmin}} \, u^T Q u  \\
	\label{seq:general optimization problem gradient condition}
	&\text{s.t. } d_{\delta^{+}} \, b_{0}(t+\delta, x+\delta (f(x)+g(x)u))\bigg|_{\delta = 0} \geq -\alpha(b_{0}(t,x)),
	\end{align}
\end{subequations}
where~\eqref{seq:general optimization problem gradient condition} ensures the satisfaction of the CBF gradient condition~\eqref{eq:cbf condition}, is numerically difficult as $ b_{0} $ is nonsmooth. Therefore, we subdivide~\eqref{eq:general optimization problem} into multiple basic optimization problems with a simplified gradient condition which can be numerically easily solved.

At first, we define some index sets that help us to describe the tree structure of the BFs constructed in Section~\ref{subsec:construction of candiate CBFs}. Recall that we denote the index set of all BFs as $ \calI $ and the index set of elementary BFs as $ \calI^{e} $. For any $ i,k\in\calI^{e} $, we define the index set 
$ \calQ^{k}_{i}(t) \!:= \!\!
{\scriptsize \begin{cases}
\!\lbrace k \rbrace & \text{if } i=k\\
\!\emptyset\!\! & \text{otherwise}
\end{cases}} $,
and for any $ k\in\calI^{e} $ and $ i\in\calI\smallsetminus\calI^{e} $ we define
\begin{align}
	\label{eq:path index set}
	\calQ^{k}_{i}(t) &\!:= \!\!
	\begin{cases}
		\!\bigcup_{i'\!\in\widetilde{\calB}_{i}\!(t)} \!\calQ_{i'}^{k}\!(t) \!\cup\! \lbrace i \rbrace\!\! & \text{if } \exists i'\!\!\in\!\widetilde{\calB}_{i}\!(t)\!\!:\! \calQ_{i'}^{k}\!(t)\!\neq\! \emptyset \!\! \\
		\!\emptyset\!\! & \text{otherwise}
	\end{cases}
\end{align}
Note that there exists \emph{at most} one $ i'\!\!\in\!\widetilde{\calB}_{i}\!(t) $ such that $ \calQ_{i'}^{k}\!(t)\!\neq\! \emptyset $ because $ \calB_{i_{1}}\cap \calB_{i_{2}}=\emptyset \;\forall i_{1},i_{2}\in\calI $ with $ i_{1}\neq i_{2} $ which is due to the tree structure of all BFs. Therefore, index set $ \calQ_{i}^{k} $ can be interpreted as a branch in the tree structure that connects the BFs with indices $ i $ and $ k $. Moreover, we call $ b_{i'}(t,x) $, $ i'\in\widetilde{\calB}_{i}(t) $ with $ i\in\calI $, an \emph{active} BF of $ b_i $ at $ (t,x)\in\calT\times\calX $ if $ b_{i}(t,x) = b_{i'}(t,x) + \gamma_{i}(t) $. Correspondingly, we define the \emph{active index set} of a BF $ b_{i}(t,x) $ as 
\begin{align}
\label{eq:active index set}
\calI^{a}_{i}(t,x) := \lbrace i' \in\widetilde{\calB}_i(t) \, | \, b_{i}(t,x) = b_{i'}(t,x) + \gamma_{i}(t) \rbrace.
\end{align}
The \emph{active elementary BF index set} for $ i\in\calI $ is defined as
\begin{align}
\label{eq:active elementary BF index set}
\begin{split}
	\calI^{e,a}_{i}\!(t,x)\!&:=\! \lbrace k\in\calI^{e} \, | \, \calQ_{k}^{i}(t) \neq \emptyset \, \wedge \\ 
	& \hspace{0.75cm} b_{i}(t,x) = h_{k}(x) + \textstyle\sum_{i'\in\calQ_{i}^{k}(t)} \gamma_{i'}(t) \rbrace
\end{split}
\end{align}
Index sets $ \calI^{a}_{i} $ and $ \calI^{e,a}_{i} $ allow to simplify $ b_{0} $ for a given time~$ t $ and state~$ x $, and in the following we take advantage of the fact that for $ k\in\calI^{e,a}_{i}(t,x) $ the BF $ b_{i} $ can be written as $ b_{i}(t,x) = h_{k}(x) + \textstyle\sum_{i'\in\calQ_{i}^{k}(t)} \gamma_{i'}(t) $. Furthermore for each pair $ k,l \in\calI^{e,a}_{0} $, there exist unique indices $ i_{kl} $, $ j_{kl} $ and $ q_{kl} $ such that $ i_{kl}, j_{kl} \in \calI^{a}_{q_{kl}} $, $ k\in\calI^{e,a}_{i_{kl}} $, $ l\notin\calI^{e,a}_{i_{kl}} $, $ l\in\calI^{e,a}_{j_{kl}} $ and $ k\notin\calI^{e,a}_{j_{kl}} $. The tree structure of the BFs and the recursive definition of the index sets allow for their efficient computation. We illustrate the meaning of the definitions by revisiting Example~\ref{ex:CBF construction}.

\vspace{-0.2cm}
\begin{example}
	Consider $ \phi_0 $ in Example~\ref{ex:CBF construction}. As all expressions are considered for the same time $ t $ and state $ x $, we omit state and time arguments. According to~\eqref{eq:path index set}, $ \calQ_{0}^{211} \!=\! \lbrace 0, 2, 21, 211 \rbrace $ specifies the indices of the branch connecting BFs with indices $ 0 $ and $ 211 $. Correspondingly, $ \calQ_{0}^{212} \!\!=\!\! \lbrace 0, 2, 21, 212 \rbrace $, $ \calQ_{211}^{211} \!\!=\!\! \lbrace 211 \rbrace $, $ \calQ_{212}^{212} \!\!=\!\! \lbrace 212 \rbrace $. Let $ \calI^{a}_{0} \!=\! \lbrace 2 \rbrace $, $ \calI^{a}_{2} \!=\! \lbrace 21 \rbrace $, $ \calI^{a}_{21} \!=\! \lbrace 211, 212 \rbrace $ be active index sets. Then, as it can be seen from Figure~\ref{fig:CBF design}, we have $ \calI^{e,a}_{0} \!=\! \calI^{e,a}_{2} \!=\! \calI^{e,a}_{21} \!=\! \lbrace 211, 212 \rbrace $, $ \calI^{e,a}_{211} \!=\! \lbrace 211 \rbrace $ and $ \calI^{e,a}_{212} \!=\! \lbrace 212 \rbrace $ according to~\eqref{eq:active elementary BF index set}. From these active elementary index sets, we can determine $ i_{kl} \!=\! 211 $, $ j_{kl} \!\!=\!\! 212 $, $ q_{kl} \!\!=\!\! 21 $ for $ k\!=\!211 $ and $ l\!\!=\!\!212 $. Loosely speaking, $ b_{q_{kl}} $ denotes the BF from which the branches leading to $ b_k $ and $ b_l $ emanate, and $ i_{kl}, j_{kl} $ are chosen such that $ \calQ_{i_{kl}}^{k} $ and $ \calQ_{j_{kl}}^{l} $ do not contain common indices. 
\end{example}
\vspace{-0.1cm}

Next, we define subsets $ \calS_{k}(t)\subseteq\calX $ on the state space as
\begin{align}
\label{eq:subset k}
\calS_k(t) := \lbrace x \,|\, k\in\calI_{0}^{e,a}(t,x) \rbrace
\end{align}
which can be equivalently written as $ \calS_k(t) = \lbrace x \,|\, b_{0}(t,x) = {b}^{k}_{0}(t,x) := h_{k}(x) + \sum_{i'\in\calQ_{0}^{k}(t)} \gamma_{i'}(t) \rbrace $. Since for all times $ t\leq T $ and all states $ x\in\calX $ there exists at least one active elementary BF $ k\in\calI_{0}^{e,a}(t,x) $, it holds that $ \calX = \bigcup_{k\in\calI^{e}} \calS_{k}(t) $. The sets $ \calS_{k}(t) $ enjoy the favorable property that $ b_{0}(t,x) $ is continuously differentiable with respect to $ x $ in the interior $ \text{Int}\,\calS_{k}(t) $ and possibly non-smooth only on $ \partial\calS_{k}(t)\cap\partial\calS_{l}(t) $ for some $ l \in\calI^{e} $. Furthermore, it holds for some $ k,l \in\calI^{e} $ that $ {b}^{k}_{0}(t,x) \!=\! {b}^{l}_{0}(t,x) $ for all $ x\in \calS_{k}(t)\cap\calS_{l}(t) $, and in particular for $ x\in \partial\calS_{k}(t)\cap\partial\calS_{l}(t) $.

Given $ x(t) \in\partial\calS_{k}(t)\cap\partial\calS_{l}(t) $ for some $ t\in\calT\smallsetminus\lbrace \beta_i \rbrace_{i\in\calI} $, we can therefore formulate the condition $ \frac{\partial}{\partial (t,x)} ({b}^{k}_{0}(t,x) - {b}^{l}_{0}(t,x)) [1, (f(x)+g(x)u)^{T}]^{T} \gtreqqless 0 $ to constrain $ u $ such that $ x(\tau)\!\in\!\calS_k(\tau) $ or $ x(\tau)\!\in\!\calS_l(\tau) $ (depending on the choice of $ \geq $ or $ \leq $) for all $ \tau\in[t,t+\delta] $ and a $ \delta \!>\! 0 $. By ensuring that $ x $ stays for some time in the interior of one of the subsets $ \calS_{k}(t) $ or $ \calS_{l}(t) $, we can take advantage of the piece-wise differentiability of $ b_0 $. In the sequel, we generalize this idea for $ x(t) \!\in\! \bigcap_{k\in\calI_{0}^{e,a}} \!\calS_{k}(t) $ and formulate an optimization problem for each $ k\in\calI^{e,a}_{0} $ with a gradient condition which is simplified in comparison to~\eqref{seq:general optimization problem gradient condition}.

For determining the boundary between $ \calS_{k}(t) $ and $ \calS_{l}(t) $, we define 
\begin{align}
	\label{eq:s_kl}
	s_{kl}(t,x) \!\!:=\!\! 
	\begin{cases}
		\!{b}^{l}_{j_{kl}}(t,x) \!-\! {b}^{k}_{i_{kl}}(t,x) \!\!\! & \text{if } b_{q_{kl}} \!\!=\! \min_{i'\in\widetilde{\calB}_{q_{kl}}} \! b_{i'}  \\
		\!{b}^{k}_{i_{kl}}(t,x) \!-\! {b}^{l}_{j_{kl}}(t,x) \!\!\! & \text{if } b_{q_{kl}} \!\!=\! \max_{i'\in\widetilde{\calB}_{q_{kl}}} \! b_{i'} 
	\end{cases}
\end{align}
where $ {b}^{k}_{i}(t,x) := h_{k}(x) + \sum_{i' \in \calQ^{k}_{i}(t)} \gamma_{i'}(t) $ for $ i\in\calI $, $ k\in\calI^{e,a}_{i}(t,x) $. Then for a given time $ t $, $ s_{kl}(t,x) = 0 $ determines those $ x $ on the boundary between $ \calS_{k}(t) $ and $ \calS_{l}(t) $, i.e., $ x\in\partial\calS_{k}\cap\calS_{l} $, $ x\in\calS_{k}\cap\partial\calS_{l} $, or $ x\in\partial\calS_{k}\cap\partial\calS_{l} $ if $ \frac{\partial s_{kl}}{\partial x}(t,x) \neq 0 $. Especially the set $ \partial\calS_{k}\cap\partial\calS_{l} $ is of interest as it can be shown that only there $ b_0 $ is non-differentiable in~$ x $. The directional derivative of $ s_{kl} $ along the trajectory of~\eqref{eq:input affine system} is given as 
\begin{align}
	\label{eq:S_kl_prime}
	s'_{kl}(t,x,u) := \dfrac{\partial s_{kl}}{\partial (t,x)}(t,x) 
	\begin{bmatrix}
		1 \\ f(x)+g(x)u
	\end{bmatrix}.
\end{align}
For an illustration of~\eqref{eq:S_kl_prime}, consider Figure~\ref{fig:illustration_sections}. If $ x(t) \in\partial\calS_{1}(t)\cap\partial\calS_{2}(t) $, $ s'_{12}(t,x,u) \geq 0 $ only admits inputs $ u $ such that $ \dot{x} = f(x) + g(x)u $ points into $ \calS_{1}(t) $, thus $ x(\tau) \in \calS_1(\tau) $ for all $ \tau\in[t,t+\delta] $ and a $ \delta > 0 $.

Using this insight, we define optimization problems with a simplified gradient constraint as
\begin{subequations}
	\label{eq:optimization problem for one section}
	\begin{align}
			\label{seq:simplified optimization problem objective}
			&u^{\ast}_k(t) = \underset{u}{\mathrm{argmin}} \; u^T Q u \\
			\label{seq:simplified optimization problem gradient condition}
			&\text{s.t. } \frac{\partial {b}^{k}_{0}}{\partial x}(t,x) (f(x)+g(x)u) + \frac{\partial {b}^{k}_{0}}{\partial t}(t,x) \geq -\alpha(b_{0}(t,x)) \\
			\label{seq:section constraint}
			&\qquad s'_{kl}(t,x,u) \geq 0 \quad \forall l\in\calI^{e,a}_{0}(t,x), \; l\neq k.
	\end{align}
\end{subequations}
for $ k\in\calI^{e,a}_{0}(t,x) $ where $ Q \in \bbR^{m\times m} $ is a positive-definite matrix and $ \alpha $ a class $ \calK $ function. If~\eqref{eq:optimization problem for one section} has no feasible solution, we set $ u^{\ast}_k = \infty $. Finally, 
\begin{align}
	\label{eq:input}
	u^{\ast}(t) = \underset{u^{\ast}_k \text{ with } k\in\calI^{e,a}_{0}(t,x)}{\mathrm{argmin}}  {u^{\ast}_k(t)}^T Q u^{\ast}_k(t)
\end{align}
is applied as control input to~\eqref{eq:input affine system}. Note that in~\eqref{eq:S_kl_prime} and~\eqref{seq:simplified optimization problem gradient condition} the time derivatives of $ s_{kl} $ and $ b_{0}^{k} $ only exist on $ \calT\smallsetminus\lbrace \beta_{i} \rbrace_{i\in\calQ_{0}^{k}} $. At times $ t\in\lbrace \beta_{i} \rbrace_{i\in\calQ_{0}^{k}} $, $ s_{kl} $ and $ b_{0}^{k} $ might be discontinuous and we consider the left sided derivative instead, i.e., $ d_{t^{-}} s_{kl} $ and $ d_{t^{-}} b_{0}^{k} $, respectively. As it can be seen from the proof in Theorem~\ref{theorem:forward invariance}, this choice is arbitrary and does not impact the invariance result.
Next, we show that there always exists a feasible solution to~\eqref{eq:input} if the class $ \calK $ function $ \alpha $ satisfies the following condition.

\vspace{-0.15cm}
\begin{assumption}
	\label{ass:class K function}
	It holds $ \alpha(b_{\text{min}}) > -\frac{\partial \gamma_{i}}{\partial t}(t) \; \forall t\in\calT, \forall i\in\calI $.
\end{assumption}
\vspace{-0.15cm}

Note that since $ \gamma_{i} $ is continuously differentiable and defined on a closed interval, $ \frac{\partial \gamma_{i}}{\partial t} $ is bounded. As the class $ \calK $ function $ \alpha $ can be freely chosen, there always exists a class $ \calK $ function $ \alpha $ such that Assumption~\ref{ass:class K function} is fulfilled. 

\begin{figure}[btp]
	\centering
	\def\svgwidth{0.6\columnwidth}
	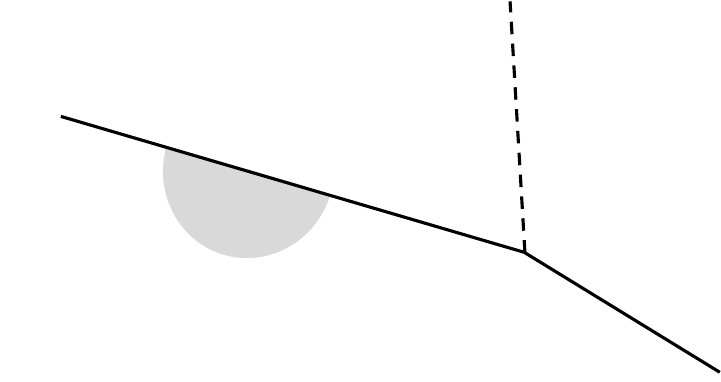
	\caption{Illustration of $ b_{0}(t,x) $ for a given $ t $, $ x\in \calX\subseteq\bbR^{2} $ and $ \phi_{0} = \calG_{[a,b]} (h_{1}(x) \geq 0) \wedge \calF_{[c,d]}( h_{2}(x) \geq 0 \vee h_{3}(x) \geq 0  ) $. The lines between the sections $ \calS_{1}, \calS_{2}, \calS_{3} $ indicate $ s_{kl} = 0 $ for $ k,l \in \calI^{e} = \lbrace 1, 2, 3 \rbrace $.}
	\label{fig:illustration_sections}
	\vspace{-0.6cm}
\end{figure}

\vspace{-0.15cm}
\begin{lemma}
	\label{lemma:existence of feasible solution}
	Let Assumption~\ref{ass:class K function} hold. For all $ (t,x) \in \calT \times \calC(t) $, there exist $ k\in\calI^{e,a}_{0}(t,x) $ such that~\eqref{eq:optimization problem for one section} has a feasible and finite solution. 
\end{lemma} 
\vspace{-0.15cm}

As it can be seen from the proof, in Assumption~\ref{ass:class K function} the class~$ \calK $ function $ \alpha $ is chosen such that no increase in $ b_{0} $ is required in~\eqref{seq:simplified optimization problem gradient condition} by varying the system's state $ x $ if $ b_{0}(t,x) \geq b_{\text{min}} $. This is especially important when $ x $ is already a maximum point of $ b_{0} $ for a given time $ t $ and a variation of $ x $ does not lead to an increase on $ b_{0} $. Loosely speaking, $ \alpha $ determines when $ b_{0}(t,x) $ is sufficiently close to zero such that the controller reacts in order ensure the invariance of the safe set~$ \calC(t) $, whereas functions $ \gamma_{i} $ determine how ``quickly'' state~$ x $ has to change. Thereby, $ \gamma_{i} $ is decisive for the magnitude of the control input~$ u $.

In the remainder of this section, we prove the forward invariance of $ \calC(t) $. Therefore, consider a solution $ \varphi\!:\! \calT \!\!\rightarrow\! \calX $ to the closed-loop system~\eqref{eq:input affine system} when control input~\eqref{eq:input} is applied. At first, we investigate the continuity properties of the control input trajectory $ u\!:\! \calT \!\!\rightarrow\! \bbR^{m} $ of the closed-loop system and show that $ \varphi $ is a Carathéodory solution. However, this is not obvious as $ u(t,x) $ is discontinuous both in $ t $ and~$ x $.

\vspace{-0.15cm}
\begin{lemma}
	\label{lemma:input continuity}
	The control input trajectory $ u^{\ast}: \calT \rightarrow \bbR^{m} $ of the closed-loop system~\eqref{eq:input affine system} with controller~\eqref{eq:optimization problem for one section}-\eqref{eq:input} is continuous a.e.
\end{lemma} 
\vspace{-0.15cm}

\vspace{-0.15cm}
\begin{corollary}
	\label{cor:continuity prop of varphi}
	A solution $ \varphi: \calT \rightarrow \calX $ to the closed-loop system~\eqref{eq:input affine system} with controller~\eqref{eq:optimization problem for one section}-\eqref{eq:input} is a Carathéodory solution. Moreover, it holds for the right sided derivative that $ d_{t^{+}} \varphi(t) = f(\varphi(t))+g(\varphi(t))u^{\ast}(t) $ for $ t\notin\lbrace \beta_{i} \rbrace_{i\in\calI} $.
\end{corollary}
\vspace{-0.1cm}

\vspace{-0.1cm}
\begin{remark}
	Solution $ \varphi $ is not necessarily unique. If multiple $ u_{k}^{\ast} $ minimize the objective in \eqref{eq:input}, then any of the inputs could be applied. Depending on which input is chosen, different solutions $ \varphi $ are obtained.
\end{remark}
\vspace{-0.1cm}

Next, we compare the controllers~\eqref{eq:general optimization problem} and~\eqref{eq:optimization problem for one section}-\eqref{eq:input}. This result facilitates the proof of forward invariance of the safe set in Theorem~\ref{theorem:forward invariance}, and allows us to compare the proposed non-smooth CBF approach to other CBF-based controllers. 

\vspace{-0.15cm}
\begin{proposition}
	\label{prop:equivalence optimization problems}
	The optimization problem~\eqref{eq:general optimization problem}
	is equivalent to \eqref{eq:optimization problem for one section}-\eqref{eq:input} for  all $ t\in\calT\smallsetminus\lbrace \beta_{i} \rbrace_{i\in\calI} $. 
\end{proposition}
\vspace{-0.15cm}


In contrast to~\eqref{eq:general optimization problem}, the constraints in $ \eqref{eq:optimization problem for one section} $ can be more easily evaluated since $ {b}^{k}_{0} $ is differentiable contrary to~$ b_0 $. Based on the previous results, we can prove forward invariance of $ \calC(t) $.

\vspace{-0.15cm}
\begin{theorem}
	\label{theorem:forward invariance}
	The control law~\eqref{eq:optimization problem for one section}-\eqref{eq:input} renders $ \calC(t) $ forward invariant for all $ t\in\calT $, and $ b_0 $ is a CBF.
\end{theorem}
\vspace{-0.15cm}



\section{Simulations}
\label{sec:simulations}

Similarly to~\cite{Lindemann2019}, we consider a multiagent system comprising three omnidirectional robots which are modeled as in~\cite{Liu2008} and use a collision avoidance mechanism as in~\cite{Lindemann2019}. The state of agent $ i $ is given as $ x_i = [p^T_i, \rho_i]^T $ where $ p_i = [x_{i,1},x_{i,2}] $ denotes its position and $ \rho_i $ its orientation; the state of all agents together is given as $ x = [x_1^T, x_{2}^{T}, x_{3}^{T}]^{T} $. The dynamics of agent $ i $ are 
\begin{align*}
	\dot{x}_{i} &= f_i(x) + 
	\begin{bmatrix}
	\cos(\rho_i) & -\sin(\rho_i) & 0 \\
	\sin(\rho_i) & \cos(\rho_i) & 0 \\
	0 & 0 & 1
	\end{bmatrix}
	(B_{i}^{T})^{-1} R_i u_i 
\end{align*}
where $ f_{i}(x) = [f_{i,1}(x), f_{i,2}(x), 0]^{T} $ with $ f_{i,k}(x) = \sum_{j=1, i\neq j}^{3} k_i \frac{x_{i,k} - x_{j,k}}{||p_{i}-p_{j}||^{2}+0.00001} $, $ k_{i}>0 $, $ B_{i} = $ {\tiny $ \begin{bmatrix}
	0 & \cos(\pi/6) & -\cos(\pi/6) \\ 
	-1 & \sin(\pi/6) & \sin(pi/6) \\
	L_i & L_i & L_i
\end{bmatrix} $}
with $ L_i = 0.2 $ as the radius of the robot body, $ R_i = 0.02 $ is the wheel radius, and $ u_i $ is the angular velocity of the wheels and serves as control input. As required, the system is input-affine and $ f_i $, $ g_i $ are continuous. Besides, we admit sufficiently large inputs to the system.

The task for the three agents comprises three parts: (1) approaching each other: $ \phi_1 := \calF_{[10,20]}(||p_1-p_2||\leq 10 \vee ||p_1-p_3||\leq 10 \vee ||p_2-p_3||\leq 10) \wedge \calG_{[20,60]}(||p_2-p_3||\leq 15) $; (2) moving to given points: $ \phi_{2} := (||p_3-[-5,-5]^{T}||\leq10)\calU_{[5,20]}(||p_1-p_2||\leq 10) \wedge \calF_{[10,20]}(||p_1-[0,30]^{T}||\leq 10) \wedge \calG_{[50,60]}(||p_{1}-[30,0]^{T}||\leq 10) \wedge (\calG_{[50,60]}(||p_2-[-30,-30]^{T}||\leq 10) \vee \calG_{[50,60]}(||p_3-[30,-30]^{T}||\leq 10)) $; and (3) staying within a defined area: $ \phi_{3} := \calG_{[0,60]}(||[p_1^{T},p_{2}^{T},p_{3}^{T}]^{T}||_{\infty} \leq 40) $. The norms $ ||\cdot|| $ and $ ||\cdot||_{\infty} $ denote the euclidean and the maximum norm, respectively. The overall task is given as the conjunction $ \phi_0 := \phi_1 \wedge \phi_2 \wedge \phi_3 $. In contrast to~\cite{Lindemann2019}, the considered task also contains disjunctions.

For the construction of the CBF $ b_0(t,x) $, all rules from Section~\ref{subsec:construction of candiate CBFs} (R0-R7) are applied; the controller is designed according to Section~\ref{subsec:controller design}. 
The simulation is implemented in Julia using Jump~\cite{Dunning2017} and run on an Intel
Core i5-10310U with 16GB RAM. The controller is evaluated with 50Hz and the control input is applied using a zero-order hold; the computation of the control input took 16ms on average.
The trajectories of the agents resulting from the simulation are depicted in Figure~\ref{fig:agent_paths}, the evolution of $ b_0 $ and the applied inputs in Figure~\ref{fig:graphs}. Since $ b_0(t,\varphi(t)) \geq 0 \; \forall t\in[0,60] $, we conclude that the specified constraints are satisfied. Besides, the inputs are indeed continuous a.e.

\begin{figure}[t]
	\vspace{0.2cm}
	\centering
	\def\svgwidth{0.7\columnwidth}
	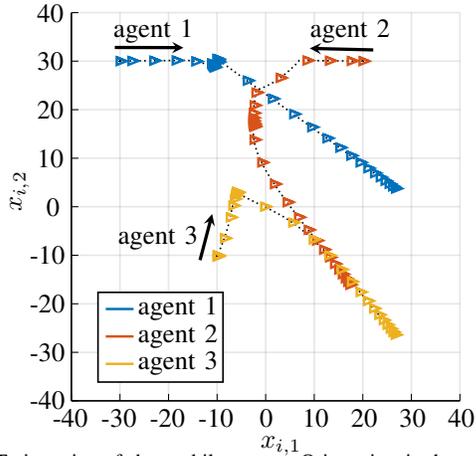\\
	\vspace{-0.4cm}
	\caption{Trajectories of the mobile agents. Orientation is denoted by the triangles' orientation.}
	\label{fig:agent_paths}
	\vspace{-0.1cm}
\end{figure}

\begin{figure}[t]
	\centering
	\subfigure[]{
			\vspace{-0.1cm}
		\def\svgwidth{0.7\columnwidth}
		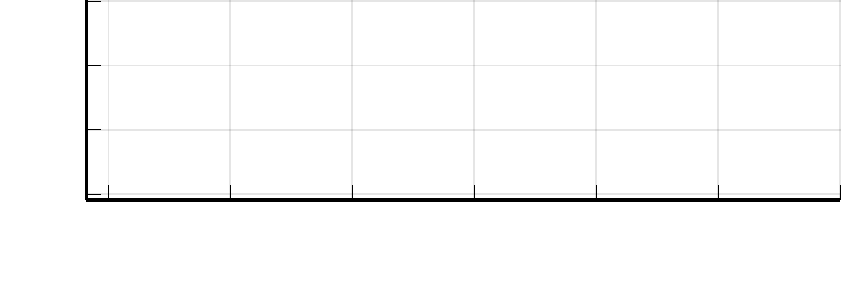
		\label{fig:b0_small}
	}
	\subfigure[]{
		\def\svgwidth{1\columnwidth}
		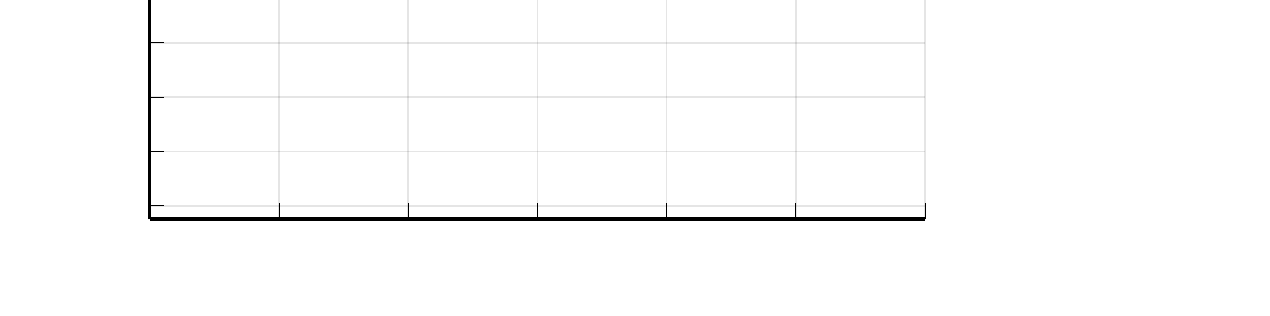
		\label{fig:inputs}
	}
	\vspace{-0.5cm}
	\caption{Control barrier function $ b_0(t,\varphi(t)) $ and control inputs $ u $ over time.}
	\label{fig:graphs}
	\vspace{-0.7cm}
\end{figure}


\vspace{-0.1cm}
\section{Conclusion}
\label{sec:conclusion}
\vspace{-0.1cm}

In this paper, we constructed a nonsmooth time-varying CBF for Signal Temporal Logic tasks including disjunctions and derived a controller that ensures their satisfaction. By using a nonsmooth approach, we avoided the problem of vanishing gradients on the CBF that occurs when employing smoothed approximations of the minimum and maximum operators. Moreover, by partitioning the state space into sections and designing an optimization problem for each of them, we could determine the respective elementary barrier function ``of relevance''. This allowed us to avoid the usage of differential inclusions in our derivation, thereby to reduce the conservativeness of our results and to apply a non-smooth approach to time-varying CBFs. 


\section*{Appendix}

\begin{proof}[Proof of Lemma~\ref{lemma:b is BF}]
	It trivially follows from the domain of $ h_i $, $ \gamma_i $ and the inverse step function $ \sigma^{-1} $, that $ b_{0} $ is defined on~$ \calT \times \calX  $. When applying minimum and maximum operators to a set of continuous and (piecewise) continuously differentiable functions, the resulting function is continuous and piecewise continuously differentiable. Thus, continuity and piecewise continuous differentiability of $ b_{0} $ in~$ x $ follow. Since additionally all functions~$ \gamma_{i} $ are continuously differentiable, and $ b_{0} $ is only discontinuous in a finite number of deactivation times, $ b_{0} $~is continuously differentiable a.e. in $ t $. 
	
	Next, we consider the discontinuities. For eventually tasks $ \phi_{i} = \calF_{[a,b]} \psi_{i'} $ (R3) and always tasks $ \phi_{i} = \calG_{[a,b]} \psi_{i'} $ (R4), it holds that $ b_{i}(t,x) $ is continuous for $ t\leq\beta_{i} $ as $ b_{i'} $ is time-independent according to R0, R1 and R2, and $ \gamma_{i} $ is continuous. If $ \phi_{i} $ is part of a conjunction or disjunction, continuity properties for $ t>\beta_{i} $ follow from the BFs constructed in R5 and R6 which are considered next. Otherwise, it holds $ b_{0}(t,x) \equiv b_{i}(t,x) \equiv 0 $ for $ t > \beta_{i} $. 
	
	For conjunctions $ \phi_i = \bigwedge_{i'\in\calB_i} \phi_{i'} $ (R5), note that $ i'' = \text{argmin}_{i'\in\calB_i} b_{i'}(\beta_{i''},x) $ for a given $ x $ if $ b_{i} $ is discontinuous at time $ \beta_{i''} $. If $ \widetilde{\calB}_i(t) \neq \emptyset $ for $ t > \beta_{i''} $, then there exists a $ i'''\in\widetilde{\calB}_i(t) $ such that $ b_{i''}(\beta_{i''},x) \leq b_{i'''}(\beta_{i''},x) $, and hence \eqref{eq:continuity condition on b} holds. 
	Otherwise if $ \widetilde{\calB}_i(t) = \emptyset $, it holds $ \beta_{i}=\beta_{i''} $. Then, either $ \phi_{i} $ is part of a conjunction or disjunction and the continuity properties for $ t>\beta_{i} $ follow recursively, or $ b_{0}(t,x) \equiv b_{i}(t,x) \equiv 0 $ for $ t > \beta_{i} $. 
	
	For disjunctions $ \phi_i = \bigvee_{i'\in\calB_i} \phi_{i'} $ (R6), observe that $ b_{i} $ are continuous a.e. in time for $ t \leq \beta_{i} $ and any discontinuities at $ t<\beta_{i} $ are due to discontinuities of $ b_{i'} $, $ i'\in\calB_i $, resulting from conjunctions (R5). Therefore, \eqref{eq:continuity condition on b} holds at all discontinuities of $ b_i $ on $ t<\beta_{i} $. Then, either $ \phi_{i} $ is again a part of a conjunction or disjunction and the continuity properties for $ t>\beta_{i} $ follow recursively, or $ b_{0}(t,x) \equiv b_{i}(t,x) \equiv 0 $ for $ t > \beta_{i} $.
\end{proof}

\begin{proof}[Proof of Theorem~\ref{thm:task satisfaction}]
	From~\eqref{seq:STL grammar predicates}, it follows directly that $ b_{i}(t,x(t)) \geq 0 \Rightarrow (x,t) \vDash \psi_i $ for $ \psi_i = p_i $ (R0).
	For the conjunction $ \phi_i = \bigwedge_{i'\in\calB_i} \phi_{i'} $ (R5), if $ b_{i}(\tau,x(\tau)) \geq 0 \; \forall \tau\in\calT $ holds, then $ b_{i''}(\tau,x(\tau)) \geq 0 $ with $ i'' = \text{argmin}_{i'\in\calB_i} b_{i'}(\tau,x(\tau)) $ for all $ \tau\in\calT $. Hence, it holds that $ b_{i'}(\tau,x(\tau)) \geq 0 \; \forall \tau\in\calT $ for all $ i'\in\calB_i $ and it follows $ x \vDash \phi_i $. A similar reasoning also holds for $ b_i $ implementing $ \psi_i = \bigwedge_{i'\in\calB_i} \psi_{i'} $ (R1).
	For disjunctions, we have to distinguish disjunctions with temporal operators (R6) and without (R2). For $ \psi_i = \bigvee_{i'\in\calB_i} \psi_{i'} $ without temporal operators (R2), it follows from $ b_{i}(t,x(t)) \geq 0 $ that $ \exists i'\in\calB_i: b_{i'}(t,x(t)) \geq 0 $ and hence $ (x,t) \vDash \psi_i $. For $ \phi_i = \bigvee_{i'\in\calB_i} \phi_{i'} $ which includes the case with temporal operators (R6), note that $ x \vDash \phi_i \Leftrightarrow \exists i'' \in\calB_i: x \vDash \phi_{i''} $. If $ b_{i}(\tau,x(\tau)) \geq 0 \; \forall \tau\in\calT $, then $ \exists i''\in\calB_i : b_{i''}(\tau,x(\tau)) \geq 0 \;\forall \tau \leq \beta_{i} $ and it follows that $ \exists i'' \in\calB_i: x \vDash \phi_{i''} $. For the \emph{eventually} operator $ \phi_{i} = \calF_{[a,b]} \psi_{i'} $ (R3), $ \gamma_{i}(t) $ is chosen such that $ \gamma_{i}(t') \leq 0  $ for some $ t' \in [a,b] $. Hence, $ b_{i}(t',x(t')) \geq 0 $ implies $ b_{i'}(t',x(t')) \geq 0 $, and the result follows from~\eqref{seq:STL grammar eventually}. For the \emph{always} operator (R4), the result follows analogously on the interval $ [a,b] $. Hence, $ b_{\phi}(t,x(t)) \geq 0 \Rightarrow x \vDash \phi $. 
\end{proof}

\begin{proof}[Proof of Lemma~\ref{lemma:existence of feasible solution}]
	At first observe that since the predicate functions $ h_{k}(x) $ are concave and the time-dependent functions $ \gamma(t) $ introduced by the temporal operators in R3 and R4 do not impact concavity with respect to $ x $, $ {b}_{0}^{k} $ in~\eqref{seq:simplified optimization problem gradient condition} is concave in $ x $ and $ b_0 $ is piecewise concave in $ x $. Piecewise concavity in $ x $ implies that for a given time $ t $, a given state $ x\in\calX $ and any direction $ d\in\bbR^{n} $ with $ 0 < ||d|| $ sufficiently small, it holds $ b_0(t,x+\varepsilon d) \geq (1-\varepsilon)b_{0}(t,x) + \varepsilon b_{0}(t,x+d) $, $ \varepsilon\in[0,1] $. As $ b_0 $ is piecewise differentiable in $ x $ by construction, we obtain by first-order Taylor approximation $ \frac{\partial b_{0}}{\partial x}(t,x) \, d \geq b_{0}(t,x+d) - b_{0}(t,x) $. Therefore, if there is no local maximum on $ b_{0} $ in $ x\in\calC(t) $ for a given $ t $, then there exists at least one direction $ d $ such that $ b_{0}(t,x+d) > b_{0}(t,x) $. Because $ \calX = \bigcup_{k'\in\calI^{e}} \calS_{k'}(t) $, there exists at least one $ k\in\calI_{0}^{e,a}(t,x) $ \ such that $ x+d\in\calS_{k}(t) $. For these $ k $, it holds $ x\notin\calH_{k} $ where $ \calH_{k} $ is the set of maxima of $ h_k $, and furthermore $ L_{g}h_{k}(x) \neq 0 $ due to Assumption~\ref{ass:predicate function}. Since $ \frac{\partial {b}_{0}^{k}}{\partial x} = \frac{\partial h_k}{\partial x} $, \eqref{seq:simplified optimization problem gradient condition} and \eqref{seq:section constraint} are satisfied for sufficiently large $ u $ and \eqref{eq:optimization problem for one section} has a feasible solution. 
	
	Next, consider the case when $ x $ is a maximum point of $ b_{0} $ for a given $ t\in\calT $. As $ \frac{\partial {b}^{k}_{0}}{\partial t}(t,x) = \frac{\partial \gamma_{i}}{\partial t}(t) $ for an $ i\in\calQ_{0}^{k}(t) $, we obtain $ \frac{\partial {b}^{k}_{0}}{\partial t}(t,x) > -\alpha(b_{\text{min}}) \geq -\alpha(b_{0}(t,x)) $ with Assumption~\ref{ass:class K function}. Because $ \frac{\partial {b}^{k}_{0}}{\partial x}(t,x) (f(x)+g(x)u) = 0 $ also holds for sufficiently large inputs $ u $ according to Assumption~\ref{ass:predicate function}, we conclude that \eqref{eq:optimization problem for one section} has a solution for any $ k\in\calI_{0}^{e,a}(t,x) $ also in this case.
	Since all terms in \eqref{seq:simplified optimization problem gradient condition} and \eqref{seq:section constraint} are finite, $ u $ is also finite. 
\end{proof}

\begin{proof}[Proof of Lemma~\ref{lemma:input continuity}]
	Let $ \varphi: \calT \rightarrow \calX $ be a solution to $ \eqref{eq:input affine system} $ that passes through a state $ x\in\calC(t) $ at time $ t $, i.e., $ \varphi(t) = x $. At first consider some time $ t\in\calT $ for which $ b_{0}(\tau,x) $ is locally Lipschitz continuous on $ \tau\in [t,t+\varepsilon] $ with an $ \varepsilon > 0 $ (this only does not hold at $ t\in \lbrace\beta_{i}\rbrace_{i\in\calI} $), and let $ u^{\ast}(t) $ be the optimal solution to~\eqref{eq:optimization problem for one section} at time $ t $. Besides, let $ \calI^{e,a}_{0}(t,x) $ be the set of indices of active elementary BFs at $ (t,x) $ as defined in~\eqref{eq:active elementary BF index set}. Since all predicate functions $ h_{i} $, $ i\in\calI^{e} $, and time-varying functions $ \gamma_{i} $, $ i\in\calI $, are continuously differentiable, it follows for $ (t,x) $ that also any function $ b_{0}^{k}(t,x) $, $ k\in\calI^{e,a}_{0}(t,x) $, is continuously differentiable in a neighborhood of $ (t,x) $. This also implies that any function $ s_{kl}(t,x) $, $ k,l\in\calI^{e,a}_{0}(t,x) $, $ k\neq l $, is continuously differentiable and $ s'_{kl}(t,x) $ in~\eqref{seq:section constraint} is continuous in all of its arguments. Moreover, note that also \eqref{seq:simplified optimization problem objective} and the right-hand side of~\eqref{seq:simplified optimization problem gradient condition} are continuous in a neighborhood of $ (t,x) $. Consequently, all constraints of \eqref{eq:optimization problem for one section} and its objective function~\eqref{seq:simplified optimization problem objective} are continuous and hence the solution $ u^{\ast}_{k}(\tau) $ to~\eqref{eq:optimization problem for one section} is continuous on $ \tau\in[t,t+\varepsilon] $ assuming that $ \varphi(\tau)\in\calS_{k}(\tau) \; \forall \tau\in[t,t+\varepsilon] $. Since~\eqref{seq:section constraint} admits only $ d\in\bbR^{n} $ such that $ \frac{\partial s_{kl}}{\partial (t,x)} \, [1,d^{T}]^{T} \geq 0 $, it holds $ \varphi(\tau)\in\calS_{k}(\tau) \; \forall \tau\in[t,t+\varepsilon] $ for at least one $ k\in\calI^{e,a}_{0}(t,x) $. Thus we conclude that $ u^{\ast} $ is continuous a.e.
\end{proof}

\begin{proof}[Proof of Corollary~\ref{cor:continuity prop of varphi}]
	Since $ f $, $ g $ are continuous, and $ u $ is continuous a.e., they are also measurable according to Lusin's theorem \cite[p.~66]{Royden2010}. Thus the Carathéodory solution $ \varphi(t) = \int_{t_0}^{t} (f(x) + g(x) u(\tau)) d\tau $ exists for any $ t_0\in\calT $. It holds $ \dot{\varphi}(t) = f(x) + g(x) u(t) $ for almost all $ t\in\calT $ and $ \varphi $ is absolutely continuous. In particular, as $ u^{\ast} $ is continuous on $ \tau\in[t,t+\varepsilon] $ for $ t\notin\lbrace \beta_{i} \rbrace_{i\in\calI} $, $ \varphi $ is differentiable on $ [t,t+\varepsilon] $. Thus, the right-sided derivative exists at $ t $ and it holds $ d_{t^{+}} \varphi(t) = f(\varphi(t))+g(\varphi(t))u^{\ast}(t) $.
\end{proof}

\begin{proof}[Proof of Proposition~\ref{prop:equivalence optimization problems}]
	Let $ t\in\calT $ be a time such that $ \varphi(t) $ is differentiable and $ t\notin \lbrace\beta_{i}\rbrace_{i\in\calI} $. Then  according to~\cite[p.~155]{Filippov1988}, it holds that
	\begin{align}
		\label{eq:directional derivative of nonsmooth function}
		\begin{split}
			\frac{d}{dt} b_0(t,&\varphi(t)) = \lim_{\delta\rightarrow 0} \frac{b_0(t+\delta,\varphi(t+\delta))-b_{0}(t,\varphi(t))}{\delta} \\
			&= \lim_{\delta\rightarrow 0} \frac{b_{0}(t+\delta,\varphi(t)+\delta\dot{\varphi}(t)) - b_{0}(t,\varphi(t))}{\delta} \\
			&= \frac{d}{d\delta} b_0(t+\delta,\varphi(t)+\delta\dot{\varphi}(t))\bigg|_{\delta=0}
		\end{split}
	\end{align}
	As the right sided derivative of $ \varphi $ exists for all $ t\in\calT\smallsetminus\lbrace \beta_{i} \rbrace_{i\in\calI} $ according to Corollary~\ref{cor:continuity prop of varphi} and it holds $ d_{t^{+}} \varphi(t) = f(x)+g(x)u^{\ast} $, the right sided derivative of $ b_{0}(t,\varphi(t)) $ exists for $ t\in\calT\smallsetminus\lbrace \beta_{i} \rbrace_{i\in\calI} $ and it follows from \eqref{eq:directional derivative of nonsmooth function} that
	\begin{align*}
		d_{t^{+}}\, b_{0}&(t,\varphi(t)) =\!\! \lim_{\delta\rightarrow 0^{+}} \frac{b_{0}(t\!\!+\!\!\delta,\varphi(t)\!\!+\!\!\delta d_{t^{+}}\!\varphi(t))\!\!-\!\! b_{0}(t,\varphi(t))}{\delta} \\
		&= d_{\delta^{+}} \, b_{0}(t+\delta, \varphi(t)+\delta (f(\varphi(t))+g(\varphi(t))u))\bigg|_{\delta = 0}.
	\end{align*}
	On the other hand, for any $ u $ there exists at least one $ k\in\calI_{0}^{e,a}(t,x) $ such that $ s'_{kl}(t,x,u) \geq 0 \;\forall l\in\calI^{e,a}_{0}(t,x),\, l\neq k $, and it holds
	\begin{align*}
		d_{t^{+}}\, b_{0}(t,\varphi(t))\! &=\! \frac{\partial {b}^{k}_{0}}{\partial x}(\!t,\!\varphi(t)\!) (f(\!\varphi(t)\!)\!\!+\!\!g(\!\varphi(t)\!)u) \!\!+\!\! \frac{\partial {b}^{k}_{0}}{\partial t}(\!t,\!\varphi(t)\!).
	\end{align*}
	Because the objectives of the minimization in~\eqref{eq:optimization problem for one section} and~\eqref{eq:input} are the same as in \eqref{eq:general optimization problem}, we conclude that optimization problem~\eqref{eq:general optimization problem}
	is equivalent to \eqref{eq:optimization problem for one section}-\eqref{eq:input} for  all $ t\in\calT\smallsetminus\lbrace \beta_{i} \rbrace_{i\in\calI} $.  
\end{proof}

\begin{proof}[Proof of Theorem~\ref{theorem:forward invariance}]
	Due to Corollary~\ref{cor:continuity prop of varphi} and Proposition~\ref{prop:equivalence optimization problems}, it holds that
	\begin{align*}
		d_{t^{+}}\, b_{0}(t,\varphi(t))\! = \! d_{\delta^{+}} \, b_{0}(t\!+\!\delta, \varphi+&\delta \dot{\varphi}(t))\bigg|_{\delta = 0} \!\!\!\!\!\! \geq \! -\alpha(b_0(t,\varphi(t)))
	\end{align*}
	for $ t\in\calT\smallsetminus\lbrace \beta_{i} \rbrace_{i\in\calI} $ and thus $ \calC(t) $ is forward invariant on any connected time interval in $ \calT\smallsetminus\lbrace \beta_{i} \rbrace_{i\in\calI} $. Furthermore, as~\eqref{eq:continuity condition on b} holds for $ b_0 $ and therefore~\eqref{eq:continuity condition of C} as well, we conclude that $ \calC(t) $ is forward invariant on $ \calT $ and $ b_{0} $ is a CBF. 
\end{proof}

\vspace{-0.15cm}
\bibliographystyle{IEEEtrans}
\bibliography{/Users/wiltz/CloudStation/JabBib/Research/000_MyLibrary}
\balance

\end{document}